%% file: main.tex
\documentclass[11pt]{article}
\usepackage{graphicx} 
\usepackage{amsthm}
\usepackage{amsmath}
\usepackage{amssymb}
\usepackage[margin=1in]{geometry}

\usepackage{thmtools}
\usepackage{thm-restate}
\usepackage[dvipsnames]{xcolor}  
\usepackage[vlined, ruled]{algorithm2e}
\usepackage[hyperfootnotes=false,hypertexnames=false,colorlinks=true,linkcolor=Maroon,citecolor=Maroon,urlcolor=Maroon]{hyperref}
\usepackage{orcidlink}                    

\usepackage[capitalize]{cleveref}

\DeclareMathOperator{\E}{\mathbb{E}}
\DeclareMathOperator{\pr}{\mathbb{P}}

\title{Decoding Balanced Linear Codes With Preprocessing}
\author{Andrej Bogdanov\footnote{{\tt abogdano@uottawa.ca}. University of Ottawa.} \and Rohit Chatterjee\footnote{{\tt rochat@nus.edu.sg}. Department of Computer Science, National University of Singapore.} \and Yunqi Li\footnote{{\orcidlink{0009-0001-3087-3190}} {\tt yunqili@comp.nus.edu.sg}. Department of Computer Science, National University of Singapore.} \and Prashant Nalini Vasudevan\footnote{\orcidlink{0000-0001-6880-795X} {\tt prashvas@nus.edu.sg}. Department of Computer Science, National University of Singapore.}}
\date{ }

\begin{document}

\maketitle

\input{preamble}

\begin{abstract}
Prange's information set algorithm is a decoding algorithm for arbitrary linear codes.  It decodes corrupted codewords of any $\mathbb{F}_2$-linear code $C$ of message length $n$ up to relative error rate $O(\log n / n)$ in $\poly(n)$ time.  We show that the error rate can be improved to $O((\log n)^2 / n)$, provided: (1) the decoder has access to a polynomial-length advice string that depends on $C$ only, and (2) $C$ is $n^{-\Omega(1)}$-balanced.

As a consequence we improve the error tolerance in decoding random linear codes if inefficient preprocessing of the code is allowed.  This reveals potential vulnerabilities in cryptographic applications of Learning Noisy Parities with low noise rate.

Our main technical result is that the Hamming weight of $Hw$, where $H$ is a random sample of \emph{short dual} codewords, measures the proximity of a word $w$ to the code in the regime of interest.  Given such $H$ as advice, our algorithm corrects errors by locally minimizing this measure.  We show that for most codes, the error rate tolerated by our decoder is asymptotically optimal among all algorithms whose decision is based on thresholding $Hw$ for an arbitrary polynomial-size advice matrix $H$.
\end{abstract}

\tableofcontents
\newpage

\input{intro}

\input{pre}

\input{algo}

\input{imposs}

\section*{Acknowledgements}

Work supported by an NSERC Discovery Grant, and by the National Research Foundation, Singapore, under its NRF Fellowship programme, award no. NRF-NRFF14-2022-0010.

\bibliographystyle{alpha}
\bibliography{refs}



\end{document}

%% file: preamble.tex
\newif\ifdraft
\drafttrue
\ifdraft
\newcommand{\note}[1]{[{\footnotesize  \textcolor{cyan}{\bf Note:} {#1}}]}
\newcommand{\ynote}[1]{[{\footnotesize \color{cyan}{\bf Yunqi:} { {#1}}}]}
\newcommand{\pnote}[1]{[{\footnotesize \color{orange}{\bf Prashant:} { {#1}}}]}
\newcommand{\anote}[1]{[{\footnotesize \color{red}{\bf Andrej:} { {#1}}}]}
\newcommand{\rnote}[1]{[{\footnotesize \color{magenta}{\bf Rohit:} { {#1}}}]}
\newcommand{\close}[0]{[{\footnotesize \color{green}{\bf Resolved.} {}}]}
\newcommand{\todo}{\colorbox{yellow}{{\rm [...]}}}
\else
\newcommand{\note}[1]{}
\newcommand{\ynote}[1]{}
\newcommand{\pnote}[1]{}
\newcommand{\anote}[1]{}
\newcommand{\rnote}[1]{}
\newcommand{\close}[0]{}
\fi


\theoremstyle{definition}
\newtheorem{definition}    {Definition} [section]
\newtheorem{approach}      {Approach}
\newtheorem{reference}     {Reference}
\theoremstyle{plain}
\newtheorem{theorem}       {Theorem}    [section]
\newtheorem{remark}        {Remark}     [section]
\newtheorem{conjecture}    {Conjecture}
\newtheorem{question}      {Question}
\newtheorem{lemma}         [theorem]    {Lemma}
\newtheorem{proposition}   [theorem]    {Proposition}
\newtheorem{claim}         [theorem]    {Claim}
\newtheorem{corollary}     [theorem]    {Corollary}
\newtheorem{construction}  [theorem]    {Construction}
\newtheorem{fact}          [theorem]    {Fact}
\newtheorem{assumption}    [theorem]    {Assumption}
\newtheorem{inftheorem}    [theorem]    {Informal Theorem}
\newtheorem{notation}      [definition] {Notation}
\newtheorem{exercise}      {Exercise}
\newtheorem*{answer}       {Answer}

\newenvironment{proof-of}[1]{\begin{proof}[\textit{Proof of #1}]}{\end{proof}}
\newenvironment{proof-sketch}{\begin{proof}[\textit{Proof Sketch}]}{\end{proof}}
\newenvironment{proof-sketch-of}[1]{\begin{proof}[\textit{Proof Sketch of #1}]}{\end{proof}}



\newcommand{\bb}       {\mathbb}
\newcommand{\ve}[1]    {\langle #1 \rangle}
\newcommand{\samp}[1]  {\overset{#1}{\leftarrow}}
\newcommand{\set}[1]   {\left\{#1\right\}}
\newcommand{\abs}[1]   {\left|#1\right|}
\newcommand{\size}[1]  {\left|#1\right|}
\newcommand{\labs}     {\left |}
\newcommand{\rabs}     {\right |}
\newcommand{\la}       {\leftarrow}
\newcommand{\ra}       {\rightarrow}
\newcommand{\zo}       {\{0,1\}}
\newcommand{\bset}     {\{0,1\}}
\newcommand{\oset}     {\{-1,1\}}
\newcommand{\Nat}      {\mathbb{N}}
\newcommand{\Int}      {\mathbb{Z}}
\newcommand{\Real}     {\mathbb{R}}
\newcommand{\Complex}  {\mathbb{C}}
\newcommand{\GF}[1]    {\mathbb{F}_{#1}}
\newcommand{\bF}       {\mathbb{F}_2}
\newcommand{\half}     {\frac{1}{2}}
\newcommand{\ceil}[1]  {\lceil#1\rceil}
\newcommand{\floor}[1] {\lfloor#1\rfloor}
\newcommand{\eps}      {\varepsilon}

\renewcommand{\Pr}[1]  {\mathrm{Pr}\left[#1\right]}
\newcommand{\Prob}[2]  {\mathop{\mathrm{Pr}}_{#1}\left[#2\right]}
\newcommand{\Exp}[1]   {\mathbb{E}\left[#1\right]}
\newcommand{\Expec}[2] {\mathop{\mathbb{E}}_{#1}\left[#2\right]}
\newcommand{\Var}[1]   {\mathrm{Var}\left[#1\right]}
\newcommand{\Vari}[2]  {\mathrm{Var}_{#1}\left[#2\right]}
\newcommand{\sbracket}[1]{\left[#1\right]}
\newcommand{\rbracket}[1]{\left(#1\right)}
\newcommand{\inner}[2] {\left\langle #1, #2 \right\rangle}
\newcommand{\cond}[1]  {\left|#1\right.}

\newcommand{\ip}[2]    {\langle #1, #2 \rangle}
\newcommand{\norm}[1]  {\|#1\|}
\newcommand{\wt}[1]    {\mathsf{wt}\left(#1\right)}
\newcommand{\sd}       {\Delta}
\newcommand{\EA}       {\mathsf{EA}}
\newcommand{\ED}       {\mathsf{ED}}
\newcommand{\sharpP}   {\#\mathsf{P}}
\newcommand{\poly}     {\mathsf{poly}}
\newcommand{\pois}     {\mathsf{Poisson}}
\newcommand{\ber}      {\mathsf{Ber}}
\newcommand{\NCP}      {\mathsf{NCP}}
\newcommand{\BNCP}     {\mathsf{BNCP}}
\newcommand{\DNCP}     {\mathsf{DBNCP}}
\newcommand{\LPN}      {\mathsf{LPN}}
\newcommand{\dist}[1]  {\mathsf{dist}\rbracket{#1}}
\newcommand{\mindist}[1]    {\mathsf{mindist}\rbracket{#1}}
\newcommand{\maxdist}[1]    {\mathsf{maxdist}\rbracket{#1}}

\newcommand{\secp}     {\lambda}
\newcommand{\Adv}      {\cal{A}} 
\newcommand{\negl}     {\mathsf{negl}}
\newcommand{\bit}      {b}
\newcommand{\Ot}       {\widetilde{O}}
\newcommand{\Omegat}   {\widetilde{\Omega}}

\newcommand{\keygen}   {\mathsf{KeyGen}}
\newcommand{\enc}      {\mathsf{Enc}}
\newcommand{\dec}      {\mathsf{Dec}}
\newcommand{\adv}      {\mathsf{adv}}

\newcommand{\SUM}      {\mathsf{SUM}}
\newcommand{\bal}      {\mathsf{Bal}}
\newcommand{\cC}       {\mathcal{C}}
\newcommand{\DC}       {\mathcal{C}^\bot}
\newcommand{\HC}       {\underline{\mathcal{C}}^{\bot}}
\newcommand{\ZC}       {Z(\mathcal{C})}
\newcommand{\hre}      {\underline{h}}

%% file: intro.tex

\section{Introduction}

Decoding corrupted codewords is a challenging algorithmic task.  Its complexity is far from being understood.  The choice of the code is crucial in this context.  Crafted codes like Reed-Solomon can be decoded optimally, all the way up to half the minimum distance.  In contrast, for random linear codes no known algorithm can substantially outperform brute force.

What is the algorithmic complexity of decoding a generic code of given rate and distance?  This question is captured by the \emph{nearest codeword problem} (NCP).

NCP takes two inputs: A linear code $C$ of message length $n$ and blocklength $m \geq n$, and a target word $w$.  (We restrict attention to codes over the binary alphabet.)  In the \emph{decision} version, the goal is to distinguish whether $w$ is close to $C$ or far from $C$.  In the \emph{estimation} version, the goal is to calculate the distance from $w$ to $C$ up to some approximation factor.  The \emph{search} version asks for the codeword closest to $w$ under the promise that $w$ is sufficiently close to $C$. 

Of the three variants, the search one is the hardest.  Decoding up to relative distance $\eta$ yields an estimation algorithm with approximation ratio at most $1/\eta$.  It also distinguishes those words that are $\eta$-close to $C$ from all others, provided that $\eta$ is within the decoding radius.

\medskip
The randomized ``information set decoding'' algorithm of Prange~\cite{Prange62} is essentially still the best asymptotically efficient one for NCP.\footnote{It was rediscovered by Berman and Karpinski~\cite{BK02} as an approximation algorithm, and partially derandomized by Alon, Panigrahy, and Yekhanin~\cite{APY09}. There have also been multiple optimizations that improve on it~\cite{BJMM2012,MO15,BothMay2018}, but these ultimately rely on similar principles and have similar asymptotic behavior.}  It finds the codeword closest to $w$ in expected time $\exp(O(\eta n))$ under the promise that $w$ is within relative distance $\eta$ of $C$.  Setting $\eta$ to $O((\log m) / n)$ yields a polynomial-time algorithm with approximation ratio $n / c \log m$ for any constant $c$.  

No asymptotic improvement to Prange's algorithm is known even in the average-case, i.e., for random linear codes $C$.  In the vanishing rate regime $m \gg n$, which will be of our main interest, most such codes are \emph{$O(\sqrt{n/m})$-balanced}:  all codewords have Hamming weight in the range $1/2 \pm \Theta(\sqrt{n/m})$, which is close to the best possible~\cite{welch74, levenshtein83, alon03, MRRW77, FT05}. 

On the hardness side, there is no polynomial-time algorithm that approximates the distance to the nearest codeword within an arbitrary constant factor unless P equals NP, or within a $O(2^{\log^{1 - \epsilon} n})$ factor unless NP has quasi-polynomial-time algorithms \cite{ABSS93, BGR25}.  In addition, there is no sub-exponential-time approximation within \emph{some} constant factor under the Exponential-Time Hypothesis \cite{BHIRW24}. 

Does this hardness stem from the choice of the code $C$ or of the corrupted codeword $w$?  The nearest codeword problem \emph{with pre-processing} is useful for studying their relative contribution.  In this problem, there is a pre-processing phase, in which the algorithm sees only the code $C$ and produces a bounded advice string $H$ that depends on $C$. There is no restriction on the complexity of computing $H$.  In the online phase, the algorithm reads $w$ and produces its answer. The estimation variant of NCP with pre-processing is known to be NP-hard to approximate to within a fixed constant factor~\cite{FM04, Reg04}, and SETH-hard to solve exactly in $2^{(1-\epsilon)n}$ time~\cite{SV19}.

\subsection*{Our results}

Our main result (\Cref{the:sea_alg}) is that pre-processing improves upon Prange's algorithm for highly balanced codes.  Under the promise that $C$ is $\beta$-balanced, the Nearest Codeword at relative distance $\eta$ can be found in time $\poly(m) \cdot \exp [ O(\eta n / \log(1/\max\{\beta, \eta\})) ]$ using advice of comparable length. 

When the code is $n^{-\Omega(1)}$-balanced, we obtain a polynomial-time algorithm with preprocessing for decoding near-codewords at relative distance $(\log n)^2 / n$.  In contrast, Prange's algorithm would take quasipolynomial time to decode at this distance.

The balance assumption is satisfied by most linear codes of rate $n^{-\Omega(1)}$ (as $n$ grows), as well as by many explicit constructions~\cite{NaorNaor1993, AlonGoldreichHastadPeralta1992, AlonBruckNaorNaorRoth1992, BenAroyaTaShma, TaShma2017}.  Is it necessary?  In \Cref{the:imp_thr} we show that it is for a certain class of algorithms with preprocessing.

Specifically, the advice in \Cref{alg:decideNCP} that solves the decision version of NCP is a matrix $H$ that depends on the code. The algorithm decides proximity to the code based on the Hamming weight of $Hw$ (modulo 2).  This weight is small if $w$ is close to the code and large if it is far.

Our \Cref{cor:lower-random} shows that for most codes $C$, any algorithm whose decision is the threshold of $Hw$ for \emph{some} matrix $H$ cannot separate words that are $\omega((\log^2 n)/n)$-close to the code from those that are $(1/2 - 3\sqrt{n/m})$-far, provided $H$ has size polynomial in $n$.  Thus even under almost extremal assumptions on the balance of $C$, the decoding distance cannot be improved. 

\Cref{the:imp_thr} is restrictive in assuming that proximity to the code is decided by a threshold of parities.  In \Cref{the:int_dist} we show that the same conclusion extends to a polynomial threshold function of parities.  While further generalizations are an interesting open question, we cannot expect to fully eliminate the complexity assumption on the distinguisher in an unconditional lower bound.  For instance, the possibility that every code can be decoded by a polynomial-size \emph{threshold of thresholds} of parities is consistent with the current sorry state of computational complexity theory.

\subsection*{Ideas and techniques}

\paragraph{The algorithm} Our advice $H$ consists of independent samples from a carefully chosen distribution $D$ on length-$m$ binary strings, conditioned on each sample being a dual codeword in $C^\perp$.  We denote (the density of) this conditional distribution by $D_C$.

A sample of $D$ is a sum $h = u_1 + \cdots + u_\ell$ of independent standard basis (one-hot) vectors in $\mathbb{F}_2^m$ chosen with replacement.  When $\ell$ is even, the all-zero codeword is in the support of $D$ and the conditional distribution $D_C$ is well-defined.

The distribution $D_C$ favors light dual codewords, with the weight controlled by the parameter $\ell$.  When $\ell$ is slightly larger than the distance of $C^\perp$, we expect $D_C$ to be concentrated on dual codewords of weight close to the minimum distance of $C^\perp$.

In \Cref{lem:bound_exp_dec} we show that for a suitable choice of $\ell$ the Fourier transform 
\[ \hat{D}_C(w) = \Expec{h \sim D_C}{(-1)^{\langle\/h, w\rangle}}
= \Expec{h \sim D}{(-1)^{\langle\/h, w\rangle}\ |\ h \in C^\perp} \]
distinguishes close from far codewords:  $\hat{D}_C(w)$ is bounded away from zero when $w$ is close to the code and extremely close to zero when it is at relative distance $1/2 - n^{-\Omega(1)}$ from all codewords and their complements.  

Our Algorithm~\ref{alg:decideNCP} decides proximity based on an empirical estimate of $\hat{D}_C(w)$.  By a large deviation bound, the efficiently computable function
\[
\tilde{D}_H(w) = \Expec{\text{row $h$ of $H$}}{(-1)^{\langle\/h, w\rangle}}
\]
uniformly approximates $\hat{D}_C(w)$ for at least one choice (in fact, for most choices) of $H$.  This $H$ is the advice provided to the algorithm.

Our search algorithm (\Cref{alg:searchNCP}) is based on one additional idea.  For sufficiently light $w$, $\hat{D}_C(w)$ is monotonically decreasing.  Thus the $1$-entries in $w$ are precisely those whose flip results in an increase of $\hat{D}_C(w)$, or of its proxy $\tilde{D}_H(w)$.  As $\hat{D}_C(w)$ is periodic modulo $C$, the errors in $w$ can be corrected by flipping those bits $w_i$ that increase the value of $\tilde{D}_H(w)$.

\paragraph{Lower bound}
To prove \Cref{the:imp_thr} we show that for any sufficiently short advice $H$, the function $\tilde{D}_H$ cannot distinguish \emph{random} codewords $w = Cx + e$ corrupted by random noise $e$ exceeding the noise rate supported by our algorithm from truly random strings $w$ in expectation.  The former are typically close to the code, while the latter are very far.  Thus it is impossible for a threshold distinguisher to tell them apart.

As a consequence, our algorithms cannot be improved by 
optimizing the advice $H$, and in particular by a different choice of distribution $D$. 
The key property of $D$ used in our analysis is that the Fourier coefficients $\hat{D}(w)$ are noticeable when $w$ is light (its relative Hamming weight is $O((\log n)^2 / n)$)
and vanishing when $w$ is heavy (see~\eqref{eq:val_fourier}).  In words, $D$ is noticeably biased under all light linear tests but almost unbiased under all heavy linear tests.

\Cref{lem:imp_mmt} states that this is optimal:  If $D$ is noticeably biased under all light tests of relative Hamming weight $\omega((\log n)^2 / n)$ then it must be somewhat biased under a heavy test of weight $(1 \pm n^{-\Omega(1)})/2$.  This consequence may be of interest in pseudorandomness, specifically the study of small-biased sets.  \Cref{lem:imp_mmt} says that one cannot ignore the contribution of light tests, which are generally easy to fool.  

\Cref{lem:imp_mmt} is a consequence of \Cref{the:imp_thr} and our proof of \Cref{the:dec_alg}.  We find this argument somewhat unsatisfying as it detours into coding theory and algorithms to establish a statement about distributions.  In \Cref{sec:optdist} we provide an alternative direct analytical proof of it.

\subsection*{Parallels and tangents to decoding in lattices}

The closest vector problem for lattices is the analogue of the nearest codeword problem for codes.  Aharonov and Regev~\cite{AR05} gave an efficient algorithm with preprocessing that approximates the distance to any lattice of dimension $n$ up to a $\sqrt{n / \log n}$ approximation factor. Liu, Lyubashevsky, and Micciancio~\cite{LLM06} gave a decoding algorithm with preprocessing for lattices under the promise that the distance between the target vector $w$ and the lattice is within $O(\sqrt{\log n / n})$ times the length of its shortest vector.

There are strong parallels between their algorithms and ours.  In the lattice setting, the analogue of our distribution $D$ is the multivariate standard normal, $H$ is a matrix of samples from $D$ conditioned on membership in the dual lattice, and the proximity decision is made based on the norm of the vector $Hw$.

One important difference is that the algorithm of Aharonov and Regev works without any assumption on the length of the dual shortest vector, which is the analogue of dual distance for lattices.  As we show in \Cref{the:imp_thr}, a bounded  dual distance assumption is necessary in the code setting.\footnote{This still leaves open whether the balance assumption used in our algorithms can be replaced by a weaker one like bounded dual distance.}  For codes, the length of the advice must be exponential in the product of the relative proximity parameter $\eta$ and the dual distance $d$. 

Interestingly, Aharonov and Regev's $\sqrt{n/\log n}$ approximation ratio for the closest vector problem with preprocessing is matched by Goldreich and Goldwasser's statistical zero-knowledge proof \cite{GG00} for the same problem.  Analogously, our $n/(\log n)^2$ approximation ratio  for the nearest codeword problem for balanced codes with preprocessing is matched by the statistical zero-knowledge proof of Brakerski et al.~\cite{BLVW19} for the same problem.  As in the lattice setting, the two algorithms are quite different. (At a technical level, the sampling of short combinations of codewords in the code and its dual, respectively, and the Fourier analysis of this process is a commonality.)

One intriguing open question is whether Aharonov and Regev's $\sqrt{n}$-approximate noninteractive (coNP) refutation for closest vector in lattices has an analogue for balanced codes.

\subsection*{Implications on learning noisy parities and cryptography}

Learning Noisy Parities (LPN)~\cite{BFKL93} is an average-case variant of NCP.  The code $C$ is random, and the decision problem is to distinguish a randomly corrupted random codeword from a uniformly random string.  The hardness of LPN has found several uses in cryptography, including public-key encryption~\cite{Ale03}, collision-resistant hashing~\cite{AHIKV17, BLVW19, YZWGL19}, and more~\cite{BLSV18,BF22,AMR25}. LPN-based schemes are attractive for their computational simplicity and plausible post-quantum security.

Some of these constructions~\cite{BLVW19, BLSV18, BF22, YZWGL19, AMR25} assume security of LPN at noise rate $O((\log n)^2 / n)$.  Our \Cref{alg:searchNCP} does not render them insecure owing to its inefficient preprocessing phase.  However, it highlights potential concerns in settings where the code may be used as a public parameter, like a hash key as in~\cite{AMR25} or a common random string as in~\cite{BLSV18}.  Such schemes may invoke security of some LPN instances where the the code $C$ is available \emph{long term}, and not generated ephemerally by running some algorithm in the scheme. This would allow an adversary to mount a longer duration attack to calculate the advice $H$ and then apply \Cref{alg:searchNCP} to break the scheme.  Our results also rule out non-uniform security:  If the adversary is only bounded in size but may otherwise depend arbitrarily on $C$ then noise of rate $O((\log n)^2 / n)$ is insecure. 

This non-uniform security model is in particular captured by the \emph{linear test framework} of Boyle et al.~\cite{BCGIKS20, CRR21, BCGIKRS22}.  They quantify the insecurity of a code $C$ by the maximum bias
\[ \max_{h \in \mathbb{F}_2^m} \Exp{(-1)^{\langle h, w\rangle}}, \]
where $w$ is a randomly corrupted random codeword of $C$. Couteau et al.~\cite[Lemma~6]{CRR21} observe that the insecurity becomes negligible when the dual distance of $C$ is sufficiently large.  The same logic underlies the proof of our lower bound \Cref{the:imp_thr}.

\subsection*{Other related work}

\paragraph{Information set decoding algorithms}  Prange's algorithm has been applied to cryptanalyze code-based schemes.  There are several concrete improvements in the high-distance regime.  None of them are asymptotic improvements in the exponent.  Specifically, the expected running time of Prange's algorithm on worst-case inputs (close to half the minimum distance) is asymptotically dominated by $2^{cn}$ with $c \approx 0.058$.  In applications, decoding at lower distances is reduced to decoding near half the minimum distance on a code of lower dimension.  

There has been considerable effort in improving this constant~\cite{MMT2011,BJMM2012,MO15,BM17}.  Both and May~\cite{BothMay2018} obtain $c \approx 0.047$.  Bernstein~\cite{Bernstein10} obtains a quadratic \emph{quantum} speedup of Prange's algorithm, and Kachigar and Tillich~\cite{KachigarTillich2017} are able to get a quantum speedup for the approaches in~\cite{MMT2011,BJMM2012}.  Ducas, Esser, Etinski and Kirshanova~\cite{DucasEsserEtinskiKirshanova2024} provide further lower-order improvements motivated by cryptographically relevant concrete parameter settings building on the approach of~\cite{MMT2011}.


\paragraph{Worst-case to average-case reductions for balanced NCP}  The work of~\cite{BLVW19} also gives a reduction from the hardness of NCP for balanced codes to that of LPN, albeit with extreme parameters. Yu and Zhang~\cite{YZ21} somewhat improve and generalize this result to another restricted kind of codes they call {\em independent} (balanced) codes. Debris-Alazard and Resch~\cite{DAR25}
give a reduction from worst case NCP on balanced codes (with inverse polynomial noise rate) to average case NCP (with noise rate inverse polynomially close to half).

%% file: pre.tex
\section{Concepts and notation}

\subsection{Notations}



\paragraph{Linear codes.}
A linear code $\cC$ is a collection of vectors $\set{C\cdot x: x\in \bF^n}$ with $C$ being the generator matrix in $\bF^{m \times n}$. By default, we assume $C$ has full rank and $m>n$. Consider any linear code $\cC$ with a generator matrix $C$, there exists a parity-check matrix $C^{\bot}\in \bF^{m\times (m-n)}$, such that $(C^\bot)^T C = 0^{(m-n)\times n}$. $C^\bot$ also generates the dual code of $\cC$, denoted by $\cC^\bot = \set{h\in \bF^m : \inner{h}{v} = 0, \forall v\in\cC }$.  


\paragraph{Weight and distance.}
We denote the Hamming weight of a vector $v \in \bF^m$ by $\wt{v}$.
For any code $\cC \subseteq \bF^m$ and vector $w\in \bF^m$, let $\dist{\cC, w} = \min_{v \in \cC} \wt{v + w}$, indicate the minimum distance of $w$ to a code $\cC$. With these notations, we are ready to define the notion of \emph{balanced codeword} and to characterize the distance of a vector from a linear code.


\begin{definition}[Balanced codeword]
    For $\beta\in[0,1]$, a length-$m$ vector $w$ is called \emph{$\beta$-balanced} if $\wt{w} \in \frac{1}{2}(1\pm \beta)m$. Correspondingly, a code $\cC$ is said to be \emph{$\beta$-balanced} if every non-zero codeword is $\beta$-balanced. 
\end{definition}
 
\begin{definition}
    For $\eta,\beta\in[0,1]$, $w$ is called \emph{$\eta$-close} to a code $\cC$ if $\dist{\mathcal{C},w} \le \eta m$; $w$ is called \emph{$\beta$-separated} from the code $\cC$, if $(w+v)$ is $\beta$-balanced for every $v\in \cC$. 
\end{definition}

\paragraph{Distributions.}  We identify discrete distributions $D$ with their probability mass functions, i.e., $D(x)$ is the probability that $x$ is the outcome when sampling from $D$.


\subsection{Nearest Codeword Problem}



The search version of \emph{nearest codeword problem} is defined as follows:

\begin{definition}[Nearest Codeword Problem]
    For $m, n\in \Nat$, $0<\eta< 1$ with $C\in \bF^{m\times n}, w\in \bF^{m}$, given input $(C,w)$ with the promise that $\dist{\mathcal{C}, w} < \eta m$, the \emph{(search) nearest codeword problem} $\NCP_{\eta}$ is to find $s$, such that $s \in \arg \min_{x\in \bF^n} \dist{C\cdot x, w}$.
\end{definition}

Beyond this general definition, we introduce a variant, called \emph{balanced nearest codeword problem}, in which the code is restricted to meet the balance property. 

\begin{definition}[Balanced Nearest Codeword Problem]
    For $m, n\in \Nat$, $0< \beta, \eta< 1/6$, consider $C\in \bF^{m\times n}, w\in \bF^{m}$, given input $(C,w)$ with the promise that $\cC$ is a $\beta$-balanced code and $\dist{\mathcal{C}, w}< \eta m$, the \emph{(search) balanced nearest codeword problem} $\BNCP_{\beta, \eta }$ is to find $\arg \min_{x\in \bF^n} \dist{C\cdot x, w}$.
\end{definition}

We note that when $\beta, \eta < 1/6$, with the promise that $w$ is $\eta$-close to a $\beta$-balanced code $\cC$, there must be a unique close codeword.
The \emph{decisional balanced nearest codeword problem} is given below.  For simplicity, we use a single parameter $\beta$ to capture both the balance of the code and the separation of the NO instances.

\begin{definition}[Decisional Balanced Nearest Codeword Problem]
    For $m, n \in \Nat$, $0< \beta,\eta< 1/6$, consider $C\in \bF^{m \times n}, w\in \bF^{m}$, given input $(C,w)$ with the promise that $\mathcal{C}$ is a $\beta$-balanced code, the \emph{decisional balanced nearest codeword problem} $\DNCP_{\beta, \eta}$ is to decide between the following two cases:
    \begin{align*}
        \mathsf{YES} &= \set{(C,w): w \text{ is } \eta\text{-close to the code } \cC};  \\
        \mathsf{NO} &= \set{(C,w): w \text{ is } \beta\text{-separated from the code } \cC}.
    \end{align*}
\end{definition}

\subsection{Random Linear Codes}

A random linear code is a code specified by a uniformly random choice of generator matrix $C \in \mathbb{F}_2^{m \times n}$.  By the Gilbert-Varshamov bound, a random linear code is $3\sqrt{n/m}$-balanced except with probability $2^{-\Omega(n)}$.

We reproduce this argument.  As each nonzero codeword of a random linear code is random, by Hoeffding's inequality it is $\beta$-balanced except with probability $2\exp(-\beta^2 m / 2)$.  By a union bound, the probability that there exists an unbalanced codeword is then at most 
\begin{equation}
\label{eq:bd_code}
(2^n-1) \cdot 2e^{-\beta^2 m/2} < 2^{n - \beta^2 m/2 + 1}. 
\end{equation}
which is $2^{-\Omega(n)}$ when $\beta = 3\sqrt{n/m}$.




%% file: algo.tex
\section{Algorithm with Preprocessing}
Our main result is \Cref{the:sea_alg}, the search algorithm with preprocessing for the balanced codeword problem.  As the algorithm for the decisional problem $\DNCP$ is easier to describe and contains the main idea, we describe it first and prove its correctness in \Cref{the:dec_alg}.  

Both of these algorithms involve preprocessing. An algorithm with preprocessing for $\NCP$ is one that, in addition to the instance $(C,w)$, is also given an advice string that is a function of $C$ (but not of $w$). This advice is not required to be efficiently computable. 

In the following statements $n$ and $m$ are the message length and blocklength of the code, respectively.

\begin{restatable}{proposition}{decAlg}
\label{the:dec_alg} Assuming $0< \beta,\eta < 1/6$, there is an algorithm with preprocessing for $\DNCP_{\beta,\eta}$ with both advice size and running time $m^2 \cdot \exp[{O(\eta n /\log(1/\beta))}]$.
\end{restatable}

\begin{restatable}{theorem}{seaAlg}
\label{the:sea_alg}
    Assuming $1/m\le \beta, \eta< 1/8$, there is an  algorithm with preprocessing for $\BNCP_{\beta,\eta}$
    with both advice size and running time $(m^4\log^2(1/\alpha)/n^2) \cdot \exp[{O(\eta n/\log(1/\alpha))}]$, where $\alpha = \beta + 2\eta$. 
\end{restatable}

In particular, when $\beta = n^{-\Omega(1)}$ and $\eta = O(\log^2n/n)$, \Cref{alg:searchNCP} (which proves \cref{the:sea_alg}) runs in polynomial time.

\begin{corollary}
    When $\beta = n^{-\Omega(1)}$, $\eta = O(\log^2 n/n)$, there is an algorithm with preprocessing for $\BNCP_{\beta, \eta}$ with advice size and running time polynomial in $n$.
\end{corollary}

Since the decisional variant $\DNCP$ reduces to $\BNCP$, \cref{the:sea_alg} is effectively more general than \cref{the:dec_alg}.  There is a gap in complexities in the regime $\eta \gg \beta$.  This gap owes to our usage of $\beta$ to represent both the balance and separation parameters in $\DNCP$.  

In \cref{sec:s2d} we present an alternative proof of \Cref{the:sea_alg} (with slightly worse advice size) by reduction to \Cref{the:dec_alg}.  While the resulting algorithm is less natural, this argument explains the deterioration from $\log 1/\beta$ in 
\Cref{the:dec_alg} to $\log 1/(\beta + 2\eta)$ in \Cref{the:sea_alg}.

\subsection{Decision Algorithm}
\label{sec:decision}


In this section, we present the decision algorithm and prove \cref{the:dec_alg}. The algorithm is as follows, where $\mathsf{Pre}$ is the preprocessing stage, and $\mathsf{Decide}$ is the algorithm itself. The algorithms are defined by parameters that we will set later in this section, based on the values of $n$, $m$, $\beta$, and $\eta$.



\begin{algorithm}[H]
    \caption{Preprocessing with parameters $(\ell, N)$}
    \label{alg:pre}
    \KwIn{Generator matrix $C\in \bF^{m \times n}$}
    \KwOut{Advice $(h_1,\dots, h_N)$}
    \SetKwProg{Samp}{Sampling}{:}{}
    \SetKwProg{Pre}{Preprocessing}{:}{}

    \Samp{$\mathsf{Samp}(C)$}{
        Sample length-$m$ unit vectors $u_1,\dots, u_\ell$ conditioned on $u_1+\dots+ u_\ell \in \cC^{\bot}$ independently randomly\;
        \Return ${u_1+\dots + u_\ell}$\;
    }

    \Pre{$\mathsf{Pre}(C)$}{
        \For{$i \in \set{1,\dots, N}$}{
            Sample $h_i \gets \mathsf{Samp}(C)$\;
        }
        \Return $h_1, \dots, h_N$\;
    }
\end{algorithm}



\paragraph{Two Important Distributions}
Fix any $n$, $m$, and $\beta$ as in the theorem statement, and any $\beta$-balanced code $C \in \bF^{m\times n}$. 
Set the parameter $\ell\in\Nat$ to some fixed value (to be determined later).  We define two distributions sampled as follows over the codespace, which will be crucial to the working of our algorithm:
\begin{itemize}
    \item Distribution $D$: sample length-$m$ unit vectors $u_1,\dots, u_\ell$ uniformly randomly, then output ${u_1 + \cdots + u_\ell}$.
    \item Distribution $D_C$: sample $h \gets D$ conditioned on $h \in \mathcal{C}^\bot$. 
\end{itemize}

We will set $\ell$ to be an even integer (among other restrictions), so that the distribution $D_C$ is well-defined. This is guaranteed by the fact that the zero vector $0^m$ lies in the dual space of any code and occurs with non-zero probability under the distribution~$D$ (so $D_C$ has nonempty support).

In the {\em preprocessing} stage {\sf Pre}, the algorithm samples vectors $(h_1, \dots, h_N)$ from distribution $D_C$, for $N\in \Nat$, which we refer to as advice. We note that since the preprocessing stage is allowed to be inefficient, it suffices to show the existence of good advice applicable to all inputs $w$ (which implies that it can be found inefficiently). We present a uniform method for obtaining such advice, that succeeds with high probability, which more than meets this bar. 


The {\em decision algorithm} {\sf Decide} performs as follows: it takes the vector $w$ and the advice $(h_1,\dots, h_N)$ associated with the code $\cC$ as input, and compares the value $\sum_{i} (-1)^{\inner{h_i}{w}}$ with a hard-coded threshold $t$ to decide whether the vector is close to or separated from the code $\cC$.  


\vspace{\baselineskip}

\begin{algorithm}[H]
    \caption{Decision algorithm with parameters $(N,t)$}
    \label{alg:decideNCP}
    \SetKwProg{Main}{Algorithm}{:}{}

    \KwIn{Vector $w\in \bF^m$ and advice $(h_1,\dots, h_N)$}
    \KwOut{$\mathsf{YES}$ or $\mathsf{NO}$, indicating the vector $w$ is closed to or separated from the code $\cC$}
    \Main{$\mathsf{Decide}(w;h_1,\dots,h_N)$ 
    }{
        \If{$ \sum_{i}(-1)^{\inner{h_i}{w}} > t\cdot N$}{
            \Return $\mathsf{YES}$\;
        }
        \Return $\mathsf{NO}$\;
    }
\end{algorithm}

\vspace{\baselineskip}

To prove correctness we analyze the Fourier transform $\hat{D}_C(w) = \E (-1)^{\langle h, w\rangle}$ and show that it is a good measure of distance to the code (\Cref{lem:bound_exp_dec}).

\paragraph{Fourier coefficient of $D_C$.}
We describe the probability mass function of $D_C$ in terms of that of $D$, which is 
\begin{equation*}
    D_C(h) = \frac{\mathbf{1}\left[h \in \DC \right] }{Z_C} \cdot D(h), 
\end{equation*}
where $Z_C$ is a normalization factor that satisfies
\begin{equation}
\label{eq:norm_zc}
    Z_C = \sum_{h \in \DC } D(h). 
\end{equation}
The Fourier coefficient of $D$ is defined by
\begin{equation*}
    \hat{D}(w) = \Expec{h \gets D}{(-1)^{\inner{h}{w}}}  = \sum_{h\in \bF^m} D(h) \cdot (-1)^{\inner{h}{w}}. 
\end{equation*}
For the distribution $D_C$, its Fourier coefficient $\hat{D}_C$ can be represented by $\hat{D}$, 
\begin{align}
\label{eq:fourier_expand}
    \hat{D}_C(w)  &= \Expec{h \gets D_C}{(-1)^{\inner{h}{w}}} \notag\\
    &= \sum_{h\in \bF^{m}} D_C(h) \cdot (-1)^{\inner{h}{w}} \notag \\
    &= \frac{1}{Z_C} \sum_{h \in \bF^m} \mathbf{1}\left[h \in \DC \right] D(h) \cdot (-1)^{\inner{h}{w}} \notag\\
    &= \frac{1}{Z_C \cdot 2^n} \sum_{v\in \mathcal{C}} \sum_{h \in \bF^{m}} D(h) \cdot (-1)^{\inner{h}{w+v}} \notag \\
    &= \frac{1}{Z_C \cdot 2^n} \sum_{v \in \mathcal{C}} \hat{D}(w+v),
\end{align}
where the fourth equality follows the fact that 
\begin{equation*}
    \mathbf{1}\left[h \in \DC\right] = \frac{1}{2^n} \sum_{v\in \mathcal{C}} (-1)^{\inner{h}{v}},
\end{equation*}
since if $h\in \DC$, for all $v\in \cC$, $(-1)^{\inner{h}{v}} = 1$; if $h\notin \DC$, there exists $v' \in \cC$, such that $\inner{h}{v'} = 1$, then 
\begin{equation*}
    \sum_{v\in \cC} (-1)^{\inner{h}{v}}= \half \sum_{v\in \cC} \rbracket{(-1)^{\inner{h}{v}} + (-1)^{\inner{h}{v + v'}}}  = 0.
\end{equation*}
Since $\hat{D}_C({0^m}) = 1$, the normalization factor is 
\begin{equation}
\label{eq:norm_factor}
    Z_C = \frac{1}{2^n} \sum_{v \in \mathcal{C}} \hat{D}(v). 
\end{equation}
Combining (\ref{eq:fourier_expand}) and (\ref{eq:norm_factor}), we obtain that 
\begin{equation}
\label{eq:hat_DC}
    \hat{D}_C(w) = \Expec{h \gets D_C}{(-1)^{\inner{h}{w}}} = \frac{\sum_{v \in \mathcal{C}} \hat{D}(v+w)}{\sum_{v \in \mathcal{C}} \hat{D}(v)}
\end{equation}

\paragraph{Fourier coefficient of $D$.}
Recall that the distribution $D$ outputs the XOR of $\ell$ independent random unit vectors over $\bF^m$,
\begin{align}
\label{eq:val_fourier}
    \hat{D}(w) = \Expec{h \gets D}{(-1)^{\inner{h}{w}}} = \Expec{u_1, \dots, u_\ell}{\prod_i (-1)^{\inner{u_i}{w}}} = \prod_i \Expec{u_i}{(-1)^{\inner{u_i}{w}}} = \rbracket{1-2\cdot \frac{\wt{w}}{m}}^\ell.
\end{align}
The last equality is obtained by
\begin{equation*}
    \Expec{u_i}{(-1)^{\inner{u_i}{w}}} = \frac{m - \wt{w}}{m} \cdot 1 + \frac{\wt{w}}{m} \cdot (-1) = 1 - 2\cdot \frac{\wt{w}}{m}.
\end{equation*}

We now turn to establishing guarantees for our algorithm. We start by showing that the function $\Expec{h \gets D_C}{(-1)^{\inner{h}{w}}}$ serves as an `ideal' distinguishing function between inputs $\eta$-close to the code, and those $\beta$-separated from the code (hence helping us decide instances of the $\DNCP$ problem). Note that this function is not necessarily efficient to compute, but we will later show that this is exactly what our preprocessing step will help us handle. 

\begin{lemma}
\label{lem:bound_exp_dec}
    For $m,n \in \Nat$, $0<\beta,\eta < 1/6$. Given any $\DNCP_{m,\beta, \eta}$ instance $(C,w)$, letting $\ell = 2 \cdot \ceil{n/\log(1/\beta)}$ be an even integer, we have that: \\ If $w$ is $\eta$-close to the code $\cC$, 
    \begin{align*}
        \Expec{h \gets D_C}{(-1)^{\inner{h}{w}}} > \frac{1}{2} (1-2\eta)^{2(n/\log(1/\beta)+1)};
    \end{align*}
    If $w$ is $\beta$-separated from the code $\cC$, 
    \begin{align*}
        \Expec{h \gets D_C}{(-1)^{\inner{h}{w}}} < 2^{-n}.
    \end{align*}
\end{lemma}

\begin{proof}
    Consider any $\DNCP_{m, \beta, \eta}$ instance $(C,w)$, where $C$ is the generator matrix of a $\beta$-balanced code $\cC$, which means that for $v = C\cdot x$ with all non-zero $x\in \bF^n$, the Hamming weight $\wt{v} \in \frac{1}{2}(1\pm \beta)m$, thus $0\le \hat{D}(v) \le \beta^l $ according to (\ref{eq:val_fourier}). As $\hat{D}(0^m) = 1$, we have
    \begin{equation}
    \label{eq:lb_sum}
        1 \le \sum_{v \in \cC} \hat{D}(v) \le 1 + (2^n-1) \cdot \beta^\ell .
    \end{equation}
    When $\ell = 2 \cdot \ceil{n/\log(1/\beta)}$, 
    \begin{equation}
    \label{eq:ub_sum}
        \sum_{v \in \cC} \hat{D}(v) \le 1+ (2^n -1)\cdot \beta^l < 1+ 2^n \cdot \beta^{2n/\log(1/\beta)} = 1+ 2^{-n}.
    \end{equation}
    
    Suppose that $(C,w)$ is a YES instance, which implies $w$ is $\eta$-close to the code $\cC$. Without loss of generality, assume $w = C\cdot x + e$ with $\wt{e} \le \eta m$, 
    \begin{align*}
        \Expec{h\gets D_C}{(-1)^{\inner{h}{w}}} = \frac{\sum_{v\in \cC} \hat{D}(v+w)}{\sum_{v\in \cC} \hat{D}(v)}
        > \frac{\hat{D}(e)}{1 + 2^{-n}}
        \ge \frac{(1-2\eta)^\ell}{1 + 2^{-n}}
        >\frac{1}{2} \rbracket{1-2\eta}^{2(n/\log(1/\beta) + 1)},
    \end{align*}
    where the first inequality follows (\ref{eq:ub_sum}) and the positivity of the Fourier coefficients; the last inequality holds for all sufficiently large $n$. When $\beta, \eta < 1/6$, $\rbracket{1-2\eta}^{2n/\log(1/\beta)} \gg 2^{-n}$.  
    
    For $w$ being a NO instance, for all $v\in \cC$, $\wt{v + w} \in \frac{1}{2}(1\pm \beta)m$, thus $\hat{D}(v+w) \le \beta^\ell$,  
    \begin{equation*}
        \Expec{h\gets D_C}{(-1)^{\inner{h}{w}}} = \frac{\sum_{v \in \mathcal{C}} \hat{D}(v+w)}{\sum_{v \in \mathcal{C}} \hat{D}(v)} \le \sum_{v \in \cC} \hat{D}(v+w) \le 2^n \cdot \beta^\ell < 2^{-n}.
    \end{equation*}
    The first inequality follows the lower bound given by (\ref{eq:lb_sum}). 
\end{proof}

\paragraph{Quality of advice.} In the following, we show that with high probability the preprocessing stage outputs {\em good} advice $(h_1,\dots, h_N) \in (\bF^m)^N$: i.e., advice on which the algorithm $\mathsf{Decide}$ yields the correct output on all inputs $w$ that satisfy the promise (for some settings of parameters $(\ell, N, t)$ depending on $n,m, \beta, \eta$). Informally, this will be advice for which inputs that are close to the code lead to an evaluation by $\mathsf{Decide}$ to a relatively high value, while inputs that are far from the code in turn lead to a significantly lower evaluation.

Recall that \cref{lem:bound_exp_dec} shows that the expected value $\Expec{h}{(-1)^{\inner{h}{w}}}$, for $h$ being sampled from $D_C$, serves as an ideal test in the sense that it is large if the vector is close to the code $\cC$; and is small when the vector is far. Using the advice, \cref{alg:decideNCP} then essentially estimates this expectation using the $N$ samples in the advice, namely $(h_1,\dots, h_N)$. Since $h_i$s are independent samples, standard concentration arguments ensure that, with high probability, the advice works for all possible $w$ (via an union bound) under the promise related to $C$. This yields the following lemma.

\begin{lemma}
\label{lem:dec_adv}
    Consider $m,n\in \Nat$ and $0< \beta,\eta < 1/6$. Given any $\DNCP_{\beta, \eta}$ instance $(C,w)$, there is a setting of parameters $(\ell, N, t)$, for which, with overwhelming probability, the preprocessing algorithm {\sf Pre} outputs advice $(h_1,\dots, h_N)$ such that for all $w$ that are either $\eta$-close to or $\beta$-separated from the code $\cC$, the decision returned by \cref{alg:decideNCP} is correct, where
    \begin{align*}
        N = 72m\cdot (1-2\eta)^{-4(n/\log(1/\beta) + 1)}.  
    \end{align*}
\end{lemma}

\begin{proof}
    Let $\ell = 2\cdot \ceil{n/\log(1/\beta)}$, by \cref{lem:bound_exp_dec}, for any $\DNCP_{m,\beta, \eta}$ instance $(C,w)$, we have either 
    \begin{equation*}
        \Expec{h \gets D_C}{(-1)^{\inner{h}{w}}} > \frac{1}{2} (1-2\eta)^{2(n/\log(1/\beta) + 1)}\text{ or }\Expec{h \gets D_C}{(-1)^{\inner{h}{w}}} < 2^{-n},
    \end{equation*}
    for $w$ being close to or separated from $C$, respectively. 
    As $\eta,\beta<1/6$, when $n$ is sufficiently large, $(1-2\eta)^{2n/\log(1/\beta)} \gg 2^{-n}$. 
    Set the threshold to be
    \begin{equation*}
        t = \frac{1}{3}(1-2\eta)^{2(n/\log(1/\beta) + 1)}
    \end{equation*}
    and let $\delta = \frac{1}{6} (1-2\eta)^{2(n/\log(1/\beta) + 1)}$. 

    For a fixed $w$ being $\eta$-close to the code $\cC$, the probability that \cref{alg:decideNCP} decides $w$ correctly with advice $(h_1,\dots, h_N)$ on the randomness of preprocessing is equivalent to 
    \begin{align*}
        \Prob{h_i}{\sum_{i} (-1)^{\inner{h_i}{w}} > t\cdot N} > \Prob{h_i}{\sum_i (-1)^{\inner{h_i}{w}} > \rbracket{\mathbb{E}_h (-1)^{\inner{h}{w}} - \delta}\cdot N}  \ge 1 -e^{-\delta^2 N/2}
    \end{align*}
    where the first inequality is obtained from $t < \mathbb{E}_h (-1)^{\inner{h}{w}} - \delta$ and the second one follows from Hoeffding's inequality; the probability is at least $(1-e^{-m})$ when $N = 2m\cdot \delta^{-2} = 72m\cdot (1-2 \eta)^{-4(n/\log(1/\beta) + 1)}$.

    Similarly, for $w$ being far from the code $\cC$, the success probability can be lower-bounded by
    \begin{align*}
        \Prob{h_i}{\sum_{i}(-1)^{\inner{h_i}{w}} \le t\cdot N} > \Prob{h_i}{\sum_i (-1)^{\inner{h_i}{w}}\le \rbracket{\mathbb{E}_h (-1)^{\inner{h}{w}} + \delta}\cdot N} \ge 1 - e^{-\delta^2 N /2}
    \end{align*}
    where the first inequality follows from $t > \mathbb{E}_h (-1)^{\inner{h}{w}} + \delta$ and the second inequality is ensured by Hoeffding's inequality. 

    Therefore, by the union bound, the probability that \cref{alg:decideNCP} outputs correctly for all possible $w$ is at least
    \begin{equation*}
        \Prob{h_i}{\forall w, \mathsf{Decide}(w; h_1, \dots, h_N)\text{ is correct}} \ge 1 - 2^m \cdot e^{-m} > 1 - 2^{-0.4m}, 
    \end{equation*}
    when $h_i$ is generated by the preprocessing algorithm. 
\end{proof}

\decAlg*

\begin{proof-of}{Proposition \ref{the:dec_alg}}
    From Lemma \ref{lem:dec_adv}, we conclude that, with overwhelming probability, $\mathsf{Pre}(C)$ returns good advice $(h_1,\dots, h_{N})$, which works for all $w$ that satisfies the promise, when $N = 72m\cdot (1-2\eta)^{-4(n/\log(1/\beta) + 1)}$. This proves that the algorithm with preprocessing works correctly. The advice size and running time of the algorithm is: $O(N\cdot m)  = m^2 \cdot \exp{O(\eta n /\log(1/\beta))}$. 
\end{proof-of}

\subsection{Search Algorithm}
In this section, we describe the search algorithm and prove Theorem \ref{the:sea_alg}. With the same preprocessing procedure applied (as in \cref{alg:pre}), the search algorithm proceeds as follows:

\vspace{\baselineskip}

\begin{algorithm}[H]
\caption{Search algorithm with parameter $N$}
\label{alg:searchNCP}
\KwIn{Vector $w \in \bF^m$ and advice $(h_1, \dots, h_N)$}
\KwOut{$\hat{x}\in \bF^n$, such that $C\cdot \hat{x}$ is the nearest codeword to $w$}
\SetKwProg{Main}{Algorithm}{:}{}


\Main{$\mathsf{Search}(w; h_1,\dots, h_N)$}{
    Set $S = \sum_k (-1)^{\inner{h_k}{w}}$\;
    \For{$i \in \set{1,\dots, m}$}{
        Set $w^{(i)} \gets w$ with $i$-th bit flipped\;
        \If{$\sum_k (-1)^{\inner{h_k}{w^{(i)}}} < S$}{
            Set $\hat{e}_i \gets 0$\;
        }
        \Else{
            Set $\hat{e}_i \gets 1$\;
        }
    }
    Set $\hat{e} = (\hat{e}_1,\dots, \hat{e}_m)$\;
    Solve the linear system $(C\cdot \hat{x} + \hat e = w)$ for $\hat x$\;
    \Return $\hat x$
}
\end{algorithm}

\vspace{\baselineskip}

Suppose the input is $w = C\cdot x + e$ with advice $(h_1,\dots, h_N)$; the algorithm proceeds to recover each bit of $e$ by testing how the value $\sum_k (-1)^{\inner{h_k}{w}}$ changes under corresponding bit flip.
Intuitively, under the distribution $D_C$ used to show \cref{lem:bound_exp_dec}, the value $\Expec{h}{(-1)^{\inner{h}{w}}}$ decreases monotonically when $w$ is farther from $\cC$. We can then toggle each coordinate of $w$ to detect coordinates $i$ for which $e_i =1$.
We show this formally below. 

\begin{lemma}
\label{lem:bound_exp_sea}
    For $m,n\in \Nat$, $1/m \le  \beta,\eta < 1/8$, consider any $\BNCP_{\beta, \eta}$ instance $(C,w)$ with the promise that $w$ is $\eta$-close to $\cC$, and let $\ell = 2\cdot \ceil{n/\log(1/(\beta+ 2 \eta))}$ be an even integer. Assume that $w = C\cdot x + e$ with $\wt{e} \le \eta m$, denote $w^{(i)}, e^{(i)}$ as the vector $w, e$ with their $i$-th bits flipped (respectively). Then, the sign of $\Delta_i$ indicates the value of the $i$-th bit $e_i$, where
    \begin{equation*}
        \Delta_i = \Expec{h \gets D_C}{(-1)^{\inner{h}{w}}} - \Expec{h \gets D_C}{(-1)^{\inner{h}{w^{(i)}}}}.
    \end{equation*}
    In particular, if $e_i =0$, $\Delta_i >\delta$; if $e_i = 1$, $\Delta_i< -\delta$; with
    \begin{equation*}
        \delta > \frac{\ell}{m} (1-4\eta)^\ell. 
    \end{equation*}
\end{lemma}

\begin{proof}
    Assume $\wt{e} = \eta'm$ for $\eta' \le \eta$.   
    
    \medskip\noindent\underline{\bf Case 1}: 
    If the $i$-th bit of $e$ is 0, the difference between $\hat{D}(e)$ and $\hat{D}(e^{(i)})$ is
    \begin{align*}
        \hat{D}(e) - \hat{D}(e^{(i)}) &= (1-2\eta')^\ell - \rbracket{1-2\left(\eta'+ \frac{1}{m}\right)}^\ell\\
        &= \rbracket{1-2\left(\eta'+\frac{1}{m}\right) + \frac{2}{m}}^\ell - \rbracket{1-2\left(\eta'+ \frac{1}{m}\right)}^\ell \\
        &\ge \frac{2\ell}{m} \cdot \rbracket{1-2\left(\eta'+ \frac{1}{m}\right)}^{\ell-1}\\
        &\ge \frac{2\ell}{m} (1-4\eta)^{\ell}.
    \end{align*}
    By (\ref{eq:hat_DC}), we have
    \begin{align*}
        \Delta_i &= \frac{\sum_{v\in \cC}\hat{D}\rbracket{v+w} - \sum_{v\in \cC}\hat{D}\left(v + w^{(i)}\right)}{\sum_{v\in \cC} \hat{D}(v)}\\
        &\ge \frac{\hat{D}\left(e\right) - \hat{D}\left(e^{(i)}\right) - \rbracket{2^n -1}\cdot \rbracket{\beta + 2\eta}^\ell}{1+ (2^n -1)\cdot \beta^\ell}. 
    \end{align*} 
    Let $\alpha = (\beta+ 2\eta)$, $\ell = 2\cdot \ceil{n/\log(1/\alpha)}$, then for sufficiently large $n$ we have
    \begin{equation*}
        2^n \cdot (\beta + 2\eta )^\ell  < 2^{-n}. 
    \end{equation*}
    When $\ell \approx 2 n/\log(1/\alpha)$ and $\beta, \eta < 1/8$, $2^{-n}\ll (2\ell/m)\cdot (1-4\eta)^\ell$ for all sufficiently large $n$. Thus,
    \begin{align*}
        \Delta_i > \frac{1}{2} \rbracket{\hat{D}(e) - \hat{D}(e^{(i)})} \ge \frac{\ell}{m} (1-4\eta)^\ell, 
    \end{align*}
which proves the claim for this case.
    
    \medskip\noindent\underline{\bf Case 2}: 
    If the $i$-th entry of $e$ is $1$, 
    \begin{align*}
        \hat{D}(e^{(i)}) - \hat{D}(e) &= \rbracket{1-2\left(\eta'- \frac{1}{m}\right)}^\ell - (1-2\eta')^\ell \\
        &= \rbracket{1 - 2\eta' + \frac{2}{m}}^\ell - (1-2\eta')^\ell \\
        &> \frac{2\ell}{m} (1-2\eta')^\ell. 
    \end{align*}
Using similar arguments as above, we can obtain 
    \begin{align*}
        - \Delta_i &= \frac{\sum_{v\in \cC}\hat{D}\rbracket{v+w^{(i)}} - \sum_{v\in \cC}\hat{D}\left(v + w\right)}{\sum_{v\in \cC} \hat{D}(v)}\\
        &\ge \frac{\hat{D}\left(e^{(i)}\right) - \hat{D}\left(e\right) - \rbracket{2^n -1}\cdot \rbracket{\beta + 2\eta}^\ell}{1+ (2^n -1)\cdot \beta^\ell} \\
        &> \frac{\ell}{m} (1-2\eta)^\ell,
    \end{align*}
which proves the claim for this case as well.

\end{proof}
We will also require the following concentration claim in our main proof. This follows from standard inequalities. 

\begin{lemma}
\label{lem:con_sea}
    For any $w_1, w_2 \in \bF^m$, if there exists $\delta>0$, such that 
    \begin{equation*}
        \Expec{h \gets D_C}{(-1)^{\inner{h}{w_1}}} - \Expec{h \gets D_C}{(-1)^{\inner{h}{w_2}}} > \delta,
    \end{equation*}
    then for $N \ge 8m \cdot \delta^{-2}$, we have that
    \begin{equation*}
        \Prob{h_1,\dots, h_N \gets D_C}{\sum_k (-1)^{\inner{h_k}{w_1}} > \sum_{k} (-1)^{\inner{h_k}{w_2}}} > 1 - 2 \exp{(-m)}.
    \end{equation*}
\end{lemma}

\begin{proof}
    By Hoeffding's inequality, 
    \begin{equation*}
        \Prob{h_1,\dots, h_N }{\sum_k (-1)^{\inner{h_k}{w_1}} > \rbracket{\mathbb{E}_h{(-1)^{\inner{h}{w_1}}} - \frac{\delta}{2}}\cdot N} > 1 - \exp{(-\delta^2 N/8)} = 1 - \exp{(-m)}, 
    \end{equation*}
    where $N = 8m \cdot \delta^{-2}$.
    Similarly, 
    \begin{equation*}
        \Prob{h_1, \dots, h_N}{\sum_k (-1)^{\inner{h_k}{w_2}} < \rbracket{\mathbb{E}_h{(-1)^{\inner{h}{w_2}}}  + \frac{\delta}{2}}\cdot N} > 1 - \exp{(-\delta^2 N/8)} = 1 - \exp{(-m)}.
    \end{equation*}
    Since $\mathbb{E}_h{(-1)^{\inner{h}{w_1}}} - \frac{\delta}{2} > \mathbb{E}_h{(-1)^{\inner{h}{w_2}}}  + \frac{\delta}{2}$, we have
    \begin{equation*}
        \Prob{h_1, \dots, h_N}{\sum_k (-1)^{\inner{h_k}{w_1}} > \sum_{k} (-1)^{\inner{h_k}{w_2}}} > 1 - 2\cdot \exp{(-m)}.  \hfill\qedhere
    \end{equation*}
\end{proof}

Finally we can show that the preprocessing step again generates good advice such that the search algorithm produces the correct output. This lemma plays a similar role to \Cref{lem:dec_adv} for the decisional setting. 

\begin{lemma}
\label{lem:sea_adv}
    For $m,n\in \Nat$, $0< \beta ,\eta< 1/8$. Given any $\BNCP_{\beta, \eta}$ instance $(C,w)$, there is a setting of parameters $(\ell, N)$, for which, with high probability, the preprocessing algorithm outputs advice $(h_1, \dots, h_N)$, such that, for all $w$ that is $\eta$-close to $\cC$, \cref{alg:searchNCP} outputs $x$ such that $C\cdot x$ is the closest codeword to $\eta$, where
    \begin{equation*}
        N = (2m^3\log^2(1/\alpha)/n^2)\cdot (1-4\eta)^{-4(n/\log(1/\alpha)+1)}, \text{ with } \alpha = \beta+ 2\eta. 
    \end{equation*}
\end{lemma}

\begin{proof}
    Consider any fixed $w = C\cdot x + e$ with $\wt{e} \le \eta m$, denote $w^{(i)}, e^{(i)}$ as the vector $w, e$ with their $i$-th bit flipped. Let $\ell = 2 \cdot \ceil{n/\log(1/\alpha)}$ and $\alpha = \beta+ 2\eta$. 
    By \cref{lem:bound_exp_sea}, when $e_i = 0$
    \begin{equation*}
        \Expec{h \gets D_C}{(-1)^{\inner{h}{w}}} - \Expec{h \gets D_C}{(-1)^{\inner{h}{w^{(i)}}}} > \frac{\ell}{m} (1-4\eta)^\ell.
    \end{equation*}
    Let $\delta = \frac{\ell}{m}(1-4\eta)^\ell$ and $N = 8m \cdot \delta^{-2} = (2m^3\log^2(1/\alpha)/n^2)\cdot (1-4\eta)^{-4(n/\log(1/\alpha)+1)}$, 
    by \cref{lem:con_sea}, the probability that $\mathsf{Search}$ in \cref{alg:searchNCP} sets $\hat{e}_i$ to be $0$ for $h_k$'s being sampled from $D_C$ can be lower-bounded by
    \begin{equation*}
        \Prob{h_1,\dots, h_N}{\hat{e}_i = 0| e_i=0} = \Prob{h_1, \dots, h_N}{\sum_k (-1)^{\inner{h_k}{w}} > \sum_k (-1)^{\inner{h_k}{w^{(i)}}}} > 1 - 2\cdot \exp{(-m)}.
    \end{equation*}
    
    When $e_i =1$, according to \cref{lem:bound_exp_sea}, 
    \begin{equation*}
         \Expec{h \gets D_C}{(-1)^{\inner{h}{w^{(i)}}}} - \Expec{h \gets D_C}{(-1)^{\inner{h}{w}}} > \frac{\ell}{m} (1-4\eta)^\ell,
    \end{equation*}
    then with the same parameters, \cref{lem:con_sea} implies that
    \begin{equation*}
        \Prob{h_1,\dots, h_N}{\hat{e}_i = 1 | e_i=1} = \Prob{h_1, \dots, h_N}{\sum_k (-1)^{\inner{h_k}{w}} \le \sum_k (-1)^{\inner{h_k}{w^{(i)}}}} > 1 - 2\cdot \exp{(-m)}.
    \end{equation*}
    For a fixed $w$, the probability that \cref{alg:searchNCP} finds the correct $x$ is at least
    \begin{equation*}
        \Prob{h_1,\dots, h_N}{\hat{x} = \arg \min_{x} \wt{C\cdot x + w}} = \Prob{h_1,\dots, h_N}{\forall i, \hat{e}_i = e_i} > 1 - 2m \cdot \exp{(-m)}.
    \end{equation*}
    
    Consider all possible $w$ that satisfies the promise, by taking the union bound, 
    \begin{equation*}
        \Prob{h_1,\dots, h_N}{\forall w, \hat{x} = \arg \min_{x} \wt{C\cdot x + w}} > 1 - 2m \cdot 2^m \cdot \exp{(-m)} >1 - m \cdot 2^{-0.4m+1}. \hfill\qedhere
    \end{equation*}
\end{proof}

\begin{proof}[Proof of~\Cref{the:sea_alg}]
    Lemma \ref{lem:sea_adv} indicates that $\mathsf{Pre}(C)$ outputs a good advice matrix $H$ with overwhelming probability, given a balanced code $\cC$, such that Algorithm $\ref{alg:searchNCP}$ works correctly for all possible input $w$, when $ N = (m^3\log^2(1/\alpha)/n^2) \cdot \exp{O(\eta n/\log(1/\alpha))}$ with $\alpha = \beta + 2 \eta$, then the advice size and running time are $O(m \cdot N) = (m^4\log^2(1/\alpha)/n^2) \cdot \exp{O(\eta n/\log(1/\alpha))}$. 
\end{proof}

\subsection{Search to Decision Reduction}
\label{sec:s2d}

\begin{proposition}
\label{prop:s2d}
    If there exists an algorithm for $\DNCP_{\beta+2\eta, \eta}$ with message length $n-1$ and block length $m$ with advice size $a$ and time complexity $t$, then there exists an algorithm for $\BNCP_{\beta, \eta}$ with advice size $an$ and time complexity $tn$.
\end{proposition}

As a consequence of \cref{prop:s2d}, the search problem $\BNCP_{\beta,\eta}$ can be solved using the decision algorithm in time $m^2n \exp O(\eta n/\log(1/\alpha))$, where $\alpha = \beta+ 2\eta$. 

\begin{proof}
    Given a $\BNCP_{\beta, \eta}$ instance $(C, w)$, where $C\in \bF^{m\times n}$ and $w\in \bF^m$ the algorithm: (1) constructs $n$ codes $C^{(i)} \in \bF^{m \times (n-1)}$ by removing the $i$-th column of $C$, (2) queries the $\DNCP_{\beta+ 2\eta, \eta}$ oracle on $(C^{(i)}, w)$ for all $i$, and (3) outputs $\hat{x}$ where
    \[ \hat{x}_i = \begin{cases} 0, &\text{if $\DNCP_{\beta+ 2\eta, \eta}(C^{(i)}, w)$ answers YES}, \\ 1, &\text{if not.} \end{cases} \]
    All codes $C^{(i)}$ are subcodes of $C$ and therefore $\beta$-balanced.
    
    Assuming $w$ is $\eta$-close to $C$, it equals $C\cdot x + e$ for some $e$ of weight less than $\eta m$.  We show that $\hat{x}$ must equal $x$.  If $x_i = 0$ then $Cx$ equals $C^{(i)}x^{(i)}$, where $x^{(i)}$ is $x$ with its $i$-th entry removed.  Therefore $w = C^{(i)}x^{(i)} + e$ is also $\eta$-close to $C^{(i)}$ and $\DNCP_{\beta+ 2\eta, \eta}$ answers YES.  
    
    If $x_i = 1$, then $w$ equals $C^{(i)}x^{(i)} + c^{(i)} + e$, where $c^{(i)}$ is the $i$-th column of $C$.  It remains to argue that $c^{(i)} + e$ is $(\beta + 2\eta)$-separated from $C^{(i)}$. This guarantees a NO answer from the oracle.  If, for contradiction, $c^{(i)} + e$ was 
    $\tfrac12(1 - \beta - 2\eta)$-close or $\tfrac12(1 + \beta + 2\eta)$-far from some codeword $c$ in $C^{(i)}$, then $c^{(i)}$ would be $\tfrac12(1 - \beta)$-close to or $\tfrac12(1 + \beta)$-far from $c$.  $c + c^{(i)}$ would then be a $\beta$-unbalanced codeword of $C$, violating the promise.
\end{proof}

%% file: imposs.tex
\section{On the Optimality of Our Algorithms}
\label{sec:imposs}
In this section, we show the optimality of the algorithm presented in the previous section. We emphasize that all our impossibility results in this section are stated for the {\em decision problem} $\DNCP_{\beta, \eta}$ with parameters set to $m = \poly(n)$ and $\beta = 3\sqrt{n/m}$. This value of $\beta$ allows an overwhelming fraction of linear codes to meet the balance requirement (see (\ref{eq:bd_code})). 

\subsection{Limits of Threshold Distinguishers}
\label{subsec:imp_thr}

Recall that \cref{alg:decideNCP} proceeds as follows: on input $w\in \bF^m$, it takes as advice a matrix $H = (h_1,\dots, h_N) \in \bF^{N \times m}$ associated with the code $\cC$, it calculates $\sum_i (-1)^{\inner{h_i}{w}}$ and compares it with a threshold $T$. We call this type of algorithm a \emph{threshold distinguisher} for code $\cC$ and refer to $N$ as its size. 

More generally, we allow a threshold distinguisher to apply some affine shift $b \in \bF^N$. That is, it computes $\sum_i(-1)^{\inner{h_i}{w}+b_i}$ where $b_i$ denotes the $i$-th coordinate of $b$.  In \Cref{sec:intervals} we show that threshold distinguishers are somewhat more powerful than they seem.

Our main lower bound shows that threshold distinguishers have limited power.


\begin{theorem}
\label{the:imp_thr}
    For $n, m, d \in \Nat$, $\beta = 3\sqrt{n/m}$, for any code $\cC$ of message length $n$, blocklength $m$, and dual distance $d$, no threshold distinguisher of size $\tfrac17 \exp[\min\{\eta d/6, 3n\}]$ correctly decides whether $w$ is $\eta$-close to or $\beta$-separated from $\cC$ for all $w$. 
\end{theorem}

All but $2^{-n/2}$ linear codes have dual distance at most $d = n/(2\log(m))$.  The reason is that there are at most $m^d$ non-zero strings of weight at most $d$, and each of them is a dual codeword with probability $2^{-n}$.  By a union bound all of them fail to be in the dual except with probability $2^{-n/2}$.  This gives the following corollary:

\begin{corollary}
\label{cor:lower-random}
    For all but a $2^{-n/2}$ fraction of linear codes $\cC$, no threshold distinguisher of size $\tfrac17 \exp\min\{(\eta n / 12\log m), 3n\}$ decides whether $w$ is $\eta$-close or $3\sqrt{n/m}$-separated from $\cC$ for all $w$.
\end{corollary}

In particular, if $m$ is polynomial in $n$ and $\eta \gg (\log n)^2 / n$, polynomial-size threshold distinguishers fail to solve $\DNCP_{\eta, \beta}$ with $\beta = 3\sqrt{n/m}$.

\begin{proof-of}{\Cref{the:imp_thr}}
    By flipping all $b_i$ if necessary, we may assume that the distinguisher $\mathsf{accept}$s when $A_{H,b}(w) \ge T$, and $\mathsf{reject}$s otherwise, where $A_{H,b}(w) = \sum_i(-1)^{\inner{h_i}{w}+b_i}$. 
    
    We argue that $A_{H,b}$ cannot ``tell apart'' a randomly corrupted random codeword from a truly random string.  Let $w = C\cdot x +e$ with random $x$ and the bits of $e$ independent $\ber(\eta/2)$.  We first show that:
    \[ \E A_{H,b}(w) \leq \E A_{H, b}(r) + N \exp(-\eta d), \]
    where $r$ is a truly random string.  By linearity of expectation, it is sufficient to argue that $\E (-1)^{\langle h, w\rangle + b}$ is $(1 - \eta)^d$-close to $\E (-1)^{\langle h, r\rangle + b}$ for every row $(h, b)$.  If $h$ is the all zero string both expectations are one.  If $h$ not a dual codeword both are zero.  The only difference comes from non-zero dual codewords.  The difference is then $(1 - \eta)^{\wt{h}} \leq (1 - \eta)^d \leq \exp(-\eta d)$.

    Now assume $A_{H, b}$ outputs at least $T$ for all strings $\eta$-close to the code and at most $T - 1$ for all $\beta$-separated strings.  By the law of total probability,
    \[ \E A_{H, b}(w) \geq T(1 - \pr(FAR)) + \E[A_{H, b}(w)\ |\ FAR] \pr(FAR) 
    \geq T - 2N \pr(FAR), \]
    where $FAR$ is the event that $w$ is $\eta$-far from $C$.  By a Chernoff bound, this has probability at most $\exp(-\eta m / 6)$.  On the other hand,
    \[ \E A_{H, b}(r) \leq (T - 1)(1 - \pr(\overline{SEP})) + \E[A_{H, b}(w)\ |\ \overline{SEP}] \pr(\overline{SEP}) \leq (T - 1) + 2N \pr(\overline{SEP}), \]
    where $\overline{SEP}$ is the event that $r$ is not $\beta$-separated from $C$.  By a union bound and a Chernoff bound, $\overline{SEP}$ had probability at most $2^n \cdot 2 \exp(-\beta^2 m/2) \leq 2 \exp(-3n)$.  In summary,
    \begin{align*}
    T - 2N\exp(-\eta m / 6) &\leq \E A_{H, b}(w) \\ &\leq \E A_{H, b}(r) + N \exp(-\eta d) \\ &\leq (T - 1) + N \exp(-\eta d) + 2N \cdot 2 \exp(-3 n). 
    \end{align*}
    It follows that 
    \[ \frac{1}{N} \leq \exp(-\eta d) + 2\exp(-\eta m/6)  + 4\exp(-3 n) \leq 7 \max\{\exp(-\eta d), \exp(-\eta m/6), \exp(-3 n)\}. \]
    As $d \leq m$, $\exp(-\eta d)$ and $\exp(-\eta m/6)$ are both bounded by $\exp(-\eta d/6)$.
\end{proof-of}

\subsection{Limits of Interval Distinguishers}
\label{sec:intervals}

Instead of a threshold, a distinguisher could potentially base its decision on some other function of the measure $A_{H, b}(w)$.  For example, it may be sensible to accept inputs whose value is close to zero and reject those whose value is far from zero, be it positive or negative.  This is a speicial case of an interval distinguisher.  In general, a \emph{$k$-interval distinguisher} partitions the range of $A_{H, b}$ into $k$ intervals and decides based on the identity of the interval that $A_{H, b}(w)$ belongs to.

\begin{corollary}
\label{the:int_dist}
For $n, m \in \Nat$, $\beta = 3\sqrt{n/m}$, for any code $\cC$ of message length $n$, blocklength $m$, and dual distance $d$, no $k$-interval distinguisher of size $\tfrac{1}{14} \exp \left[\min\{\eta d/6, 3n\}/(k-1)\right]$ correctly decides whether $w$ is $\eta$-close to or $\beta$-separated from $\cC$ for all $w$.    
\end{corollary}

We prove it by reducing an interval distinguisher to a threshold distinguisher.

\begin{lemma}
Every $k$-interval distinguisher of size $N$ can be simulated by a threshold distinguisher of size at most $(2N)^{k-1}$.
\end{lemma}

\begin{proof}
    For every partition of the line into $k$ intervals there is a polynomial $p$ of degree $k-1$ that alternates sign among the intervals.  Given a $k$-interval distinguisher $A_{H,b}$ we construct a threshold distinguisher ${A}_{\hat{H}, \hat{b}}$ so that
    \begin{equation}
    \label{eq:poly_con}
        {A}_{\hat{H},\hat{b}}(w) = p(A_{H,b}(w)). 
    \end{equation}
    The value ${A}_{\hat{H},\hat{b}}(w)$ is positive precisely when $A_{H, b}(w)$ falls into a positive interval.  The polynomial $p$ factorizes as
    \begin{equation*}
        p(z) = \prod_{i=1}^{k-1} (z - a_i),
    \end{equation*}
    where the breakpoints $a_i$ are integers between $-N$ and $N$. 

    The advice $\hat{H}, \hat{b}$ will be constructed by ``applying'' $p$ to $H, b$.  To do so we describe how to subtract a constant from a threshold distinguisher, and how to multiply two threshold distinguishers.

    To subtract a constant $a$ from $H, b$ we pad $H$ with $\abs{a}$ zero rows and pad $b$ with $\abs{a}$ ones if $a > 0$ and $\abs{a}$ zeros if $a < 0$.  The resulting distinguisher $\hat{H}, \hat{b}$ satisfies
    \[ A_{\hat{H}, \hat{b}}(w) = A_{H, b}(w) - a. \]
    
    To multiply two distinguishers $H, b$ and $H', b'$, we create a new distinguisher $\hat{H}, \hat{b}$ whose rows are the XORS of all pairs of rows of $(H, b)$ and $(H', b')$.  Then
    \[ A_{\hat{H}, \hat{b}}(w) = \sum_{h, h'} (-1)^{\langle h + h', w\rangle + b + b'} 
    = \Bigl(\sum\nolimits_h (-1)^{\langle h, w\rangle + b}\Bigr)\Bigl(\sum\nolimits_{h'} (-1)^{\langle h', w\rangle + b'}\Bigr) = A_{H, b}(w) \cdot A_{H', b'}(w). \]
    
    Applying $N$ subtractions and $k - 2$ multiplications we obtain a distinguisher satisfying~\eqref{eq:poly_con}.  Subtraction affects size by at most an additive $N$.  As the original distinguisher has size $N$ it at most doubles it.  Multiplication results in a distinguisher whose size is the product of its parts.
    The size of the final distinguisher is therefore at most $(2N)^{k-1}$.
\end{proof}

\subsection{Optimality of our Distribution}
\label{sec:optdist}

The specific measure $D$ over $m$-bit strings defined in \Cref{sec:decision} is key for our analysis.  By \eqref{eq:val_fourier}, $w$ has nonnegligible bias against tests of relative weight $O((\log n)^2 / n)$ and tiny bias against all $n^{-\Omega(1)}$ balanced tests (assuming $C$ is $n^{-\Omega(1)}$-balanced).  (In contrast, a symmetric product measure with the same expected Hamming weight has negligible bias beyond $w = O(\log n / n)$.)

Is there a better choice of $D$ that remains biased against tests of weight $\omega((\log n)^2 / n)$ yet remains pseudorandom against all balanced tests?  Such a $D$ would have improved the decoding radius $\eta$ in \Cref{the:dec_alg}.  The improved algorithm would have then stood in contradiction to our lower bound \Cref{the:imp_thr}.  We summarize the conclusion:

\begin{lemma}
\label{lem:imp_mmt}
Assume $m > c (\ln 2/\beta) (\ln 1/\gamma) / \beta^2 \eta$ for some absolute constant $c$.  For any distribution $D$ over $\{-1, 1\}^m$, if $\hat{D}(w) \geq \gamma$ for all $w$ of weight at most $\eta m$ (for even $\eta m$ and $\eta \leq 1/4$), there must exist a $\beta$-balanced $w$ for which
\[ \hat{D}(w) \geq \exp \left[- \frac{\ln 2/\beta \cdot \ln 1/\gamma}{\eta}\right]. \]
\end{lemma}

In our application to decoding, when $\beta = n^{-\Theta(1)}$ and $\gamma = n^{-\Theta(1)}$, the lower bound reads $\hat{D}(w) \geq \exp -O((\log n)^2 / \eta)$.  When $\eta$ is $\Theta((\log n)^2 / n)$ some $\beta$-balanced test has bias at least $\exp -O(n)$, matching~\eqref{eq:val_fourier}.

\Cref{lem:imp_mmt} is a statement about distributions over $m$-bit strings that makes no reference to codes and algorithms.  Qualitatively, it says that if a distribution has noticeable bias against all light tests, then it must have some bias against some balanced test.  We give a direct analytical proof of this.

\begin{proof-of}{Lemma \ref{lem:imp_mmt}}
Let $X$ be a sample of $D$ and $t = \eta m$.  Let $I$ be a uniformly random index in $\{1, \dots, m\}$.  By monotonicity of moments, for every $K \geq t$,
\[ \E (\E X_I | X)^K \geq \bigl(\E (\E X_I | X)^t\bigr)^{K/t}. \]
We can expand $\E (\E X_I | X)^t$ as $\E X_{I_1} \cdots X_{I_t}$, where $I_1, \dots, I_t$ are independent replicas of $I$.  By our assumption on the light tests, for every fixing of $I_1, \dots, I_t$ this expression is at least $\gamma$.  Therefore
\[ \bigl(\E (\E X_I | X)^t\bigr)^{K/t} \geq \gamma^{K/t}. \]
Let $K$ be a Poisson random variable of rate $\lambda m$ (to be determined).  Averaging over $K$ we obtain
\[ \E (\E X_I | X)^K \geq \E \gamma^{K/t} \cdot 1(K \geq t) \geq \E \gamma^{K/t} - \pr(K < t). \]
as $\gamma \leq 1$.  Using the formula for the Poisson moment generating function,
\[ \E \gamma^{K/t} = \exp \lambda m(\gamma^{1/t} - 1) \geq \gamma^{\lambda m/t} = \gamma^{\lambda / \eta}. \]
and so
\begin{equation}
\label{eq:imp1}
 \E (\E X_I | X)^K \geq \gamma^{\lambda / \eta} - \pr\bigl(\text{Poisson}(\lambda m) < \eta m\bigr). 
\end{equation}
The left-hand side also expands into a product of a $K$ independent replicas $\E X_{I_1} \cdots X_{I_K}$.  Let $N_i$ be the number of times $X_i$ occurs in this expression, i.e. $N_i$ is the number of indices $j$ for which $I_j = i$.  By the additivity of Poisson random variables, the $N_i$ are independent Poissons of rate $\lambda$.  Therefore
\begin{equation}
\label{eq:imp2}
\E (\E X_I | X)^K = \E X_1^{N_1} \cdots X_m^{N_m} = \E X_1^{\oplus N_1} \cdots X_m^{\oplus N_m} = \E \hat{D}(\oplus N),
\end{equation}
where $\oplus N_i = N_i \bmod 2$ and $\oplus N = (\oplus N_1, \dots, \oplus N_m)$, because $X_i^2 = 1$.  The random variables $\oplus N_i$ inherit their independence from $N_i$. Marginally they are Bernoulli of bias
\[ \E (-1)^{\oplus N_i} = \E (-1)^{N_i} = \exp -2\lambda. \]
again using the formula for the Poisson moment generating function.  

We choose $\lambda$ so that $\beta = 2\exp -2\lambda$.  By Chernoff bounds, $\oplus N$ is $\beta$-balanced except with probability $\exp -\Omega(\beta^2 m)$.  By the total probability theorem,

\[ \E \hat{D}(\oplus N) \leq \max_{\text{$w$ is $\beta$-balanced}} \hat{D}(w) + \exp -\Omega(\beta^2 m). \]
Plugging into~\eqref{eq:imp1} and~\eqref{eq:imp2},
\[ 
\max_{\text{$w$ is $\beta$-balanced}} \hat{D}(w) \geq \gamma^{(\ln 2/\beta) /2\eta} - \exp -\Omega(\beta^2 m) - \pr\bigl(\text{Poisson}(\lambda m) < \eta m\bigr). \]
Under our assumptions on $m$ and $\eta$ and the Poisson large deviation inequality~\cite[Theorem~A.8]{canonne} the leading term dominates and gives the desired bound.
\end{proof-of}